\documentclass[11pt]{llncs}
\usepackage{amsmath, amssymb, amsfonts}
\usepackage[pdftex]{graphicx}
\usepackage[english]{babel}
\usepackage{multirow,a4}
\usepackage{multirow}
\usepackage{array}

\newcommand{\abc}[3]{$\mbox{$\textup{#1}|#2|#3$}$}

\newcommand{\poa}{\textup{PoA}}
\newcommand{\thickhline}{\noalign{\hrule height 0.8pt}}
\newcommand{\eoe}{\hfill$\triangleleft$}

\begin{document}
\title{The Quality of Equilibria for Set Packing Games}
\author{Jasper de Jong, Marc Uetz}
\institute{University of Twente, Enschede, The Netherlands\\
{\tt \{j.dejong-3,m.uetz\}@utwente.nl}}
\maketitle

\begin{abstract}
We introduce set packing games as an abstraction of situations in which $n$ selfish players select subsets of a finite set of indivisible items, and analyze the quality of several equilibria for this class of games. Assuming that players are able to approximately play equilibrium strategies, we show that the total quality of the resulting equilibrium solutions is only moderately suboptimal. 
Our results are tight bounds on the price of anarchy for three equilibrium concepts, namely Nash equilibria, subgame perfect equilibria, and an equilibrium concept that we refer to as $k$-collusion Nash equilibrium.  
\end{abstract}

\section{Introduction}\label{sec:intro}
The set packing problem is one of Karp's 21 $\mathsf{NP}$-complete problems~\cite{Karp}; it is problem [SP3] in \cite{GJ1979}. In set packing, the task is to select from a given collection $\mathcal{S}$ of subsets of some finite universe of items $J$, a collection of $k$ disjoint subsets, for a given number $k$. In the weighted optimization version of the problem, each subset $S\in\mathcal{S}$ has a weight $w(S)$, and the goal is to find disjoint subsets of maximum total weight. Set packing is NP-complete~\cite{GJ1979}, and with respect to the approximability of the weighted optimization version, see, e.g.~\cite{Chandra}.

We here propose and analyze a game theoretic variation of the maximum weight set packing problem, which is defined as follows. We have set of $n$ players, indexed~$i\in\{1,\dots,n\}$. Each player is equipped with a downward closed collection $\mathcal{S}_i$ of subsets of a finite ground set $J$. $\mathcal{S}_i$ are the subsets that are feasible for player $i$, and then $\mathcal{S}=\cup_{i=1}^n \mathcal{S}_i$ is the collection of feasible subsets of all players. 
Each item $j\in J$ has a weight $w_j$, and the objective of any player~$i$ is to select a subset $S\in \mathcal{S}_i$ maximizing $w(S)=\sum_{j\in S}w_j$. Any of the items $j\in J$, however, can only be selected by at most one of the players. In that situation, we define a (pure strategy) Nash equilibrium as a selection of subsets $S_i$, one for every player $i=1,\dots,n$, so that $S_i\cap S_k =\emptyset$ for any two players $i\neq k$, and for each player $i$, none of the subsets $T_i\in \mathcal{S}_i$ with $T_i\subseteq J\setminus (\cup_{k\neq i}) S_k$ has a value $w(T_i)$ larger than $w(S_i)$. 
In words, given the items selected by other players, among the feasible subsets still available to player $i$, $S_i$ is the one that maximizes total value. Note the following peculiarity: When considering set packing games as $n$-player strategic form games with strategy spaces $\mathcal{S}_i$ per player, and $(S_1,\dots,S_n)$ is a strategy profile, we have to declare the payoff equal to $-\infty$ for all players $i$ and $k$ with $S_i\cap S_k\neq\emptyset$, in order to always guarantee that an equilibrium outcome corresponds to a set packing. As a consequence of this definition, once the strategies of all players except $i$ are fixed at some strategy vector $S_{-i}=(S_1,\dots,S_{i-1},S_{i+1},\dots,S_n)$, any rational choice by player $i$, i.e., with payoff $>-\infty$, does not affect the payoff of the other players. 

Our interest goes into the quality of several types of equilibria for this class of games, which to the best of our knowledge has not been addressed so far. For the entire paper, we measure the quality of an (equilibrium) solution by the total value of all selected items, or equivalently, the sum of all players' selected values $\sum_{i=1}^n w(S_i)$. The question is by how much an equilibrium solution falls behind an optimal solution that could be computed by some central authority.
 For a maximization problem as the one considered here, recall that the \emph{price of anarchy} \cite{Papadimitriou2001,Koutsoupias} denotes the ratio of the value of an optimal solution over the value of an equilibrium solution.  We analyze the {price of anarchy}  for three different equilibrium concepts, namely Nash equilibria, subgame perfect equilibria (defined by Selten~\cite{Selten65}) of a sequential version of the game, and a third equilibrium concept that we refer to as $k$-collusion Nash equilibria, as also 
 defined by Hayrapetyan et al.~\cite{Hayrapetyan}.
Because the combinatorial problems of an equilibrium play of any player could be NP-hard in general, we consider $\alpha$-approximate versions for all three equilibrium concepts, for any $\alpha\ge 1$. The idea of approximate equilibria is by now a widely accepted concept with different variations. Already Roughgarden and Tardos \cite{Roughgarden} consider it for network routing games; see also \cite{SkopalikVoecking2008} for hardness results in the context of congestion games.
Our price of anarchy bounds are tight for all $\alpha$-approximate versions.

Our original motivation to look into this class of games is a subclass of set packing games, namely throughput scheduling games. It is precisely this subclass of set packing games that we have studied in an extended abstract underlying this paper \cite{deJong2013}. In throughput scheduling, studied e.g.\  in \cite{BarNoy2001,BermanDasgupta2000} from the algorithmic perspective, the set $J$ corresponds to a set of non-preemptive jobs, each with a release time $r_j$, due date $d_j$, and a weight $w_j$. Each player has one or several machines in order to process jobs. In the most general setting, the machines can be unrelated, meaning that the processing time of any job  may depend on the machine $\ell$ it is processed on, and the $\ell\times m$ matrix $(p_{\ell j})$ of processing times on machines can have rank $>1$. A subset $S_i$ of jobs is then feasible for player $i$ if there exists a schedule of the jobs in $S_i$ on the set of machines of player $i$, so that each job can be processed in the time window $[r_j,d_j]$. Obviously, the set of jobs feasible for player $i$ is then  downward closed.

Our contribution is summarized as follows. If all players are able to play $\alpha$-approximate Nash equilibria, the price of anarchy for set packing games equals $\alpha+1$.  We also show that $\alpha$-approximate subgame perfect equilibria of a sequential version of set packing games have a price of anarchy equal to $\alpha+1$, but for the special case of symmetric set packing games (to be defined later),
subgame perfect equilibria yield an improved price of anarchy of ${\sqrt[\alpha]{e}}/({\sqrt[\alpha]{e}-1})$, which is tight, too.  Finally, we define ($\alpha$-approximate) $k$-collusion Nash equilibria. They have been defined before by Hayrapetyan et al.~\cite{Hayrapetyan} in the context of congestion games to study the price of collusion, and constitute a generalization of $k$-strong Nash equilibria \cite{Aumann1959,Andelman2009}.  The simple idea is that up to $k$ players may collude and are allowed to use any profit sharing protocol among themselves, hence can be thought of as acting like a single player. Specifically, an $n$-collusion Nash equilibrium is then just another name for an optimal solution.
For that equilibrium concept, and when players are assumed to be able to play $\alpha$-approximate $k$-collusion Nash equilibria,
 we derive a tight bound on the price of anarchy equal to $\alpha+(n-k)/(n-1)$. 

\section{Motivation \& Related Work}

Our motivation to study set packing games is to understand the performance of decentralized service systems where items are posted, e.g.\ on an internet portal, and service providers can select these jobs on a take-it-or-leave-it basis. The problem  can be seen as a stylized version of coordination problems that appear in several application domains. We give three examples.
(1) When operating microgrids for decentralized energy production, the goal is to consume locally produced energy as much as possible.  Here, the items are the operation of appliances in households (e.g.\ loading a car battery) which come with a time window and a certain monetary value. Players, on the other side, are intermediaries or local energy producers that want to maximize the total value of items than can be accepted given a profile of available energy; see, e.g.~\cite{BakkerEtAl2010,MolderinkEtAL2010} for more context.
 (2) In cloud computing, service providers such as Google provide an infrastructure service. Here, the items are computational tasks to be distributed over data- and computing centers. The aim of a federated cloud computing environment, e.g.~\cite{Buyya}, is to ``coordinate load distribution among different cloud-based data centers in order to determine optimal location for hosting application services''.  (3) In private car sharing portals like e.g.\ Tamyca~\cite{tamyca}, items are car rental requests for a certain time period, and such a request comes with a given price. Car owners in the vicinity can select such requests from the portal and rent out their car(s).
Stripping off some of the potentially complicating practical features from these applications, exactly yields the type of set packing problems that we we address here.

The overall conclusion of the analysis of equilibria that we provide here, in the light of these applications, is that the loss of efficiency caused by the lack of centralized distribution of items is only very moderate.

As to related work on set packing games, we are not aware of publications that have addressed this specific problem before. Much of the work in algorithmic game theory addresses auctions, congestion games or other types of scheduling and load balancing games. 
One distinguishing feature of set packing games is that players, e.g.\ machines in throughput scheduling games, select items and not vice versa. Of course, other models also exist where e.g.\ machines are the set of players, most prominently the task scheduling problem of the seminal publication on algorithmic mechanism design by Nisan and Ronen \cite{Nisan}. There, however, the strategy spaces of the players are the times required to perform all tasks, not the selection of tasks.

Moreover, as discussed already above, a distinguishing feature of set packing games, when compared e.g.\ with congestion games or many other machine scheduling games, is the fact that by the specific payoff structure that we impose, players other than player~$i$ influence the \emph{availability} of strategies from the strategy set $\mathcal{S}_i$ for player $i$, and the strategy then chosen by player $i$, when rational, does not affect other players anymore. It is this feature that admittedly appears somewhat special, yet in a sequential version of the game where players select items one after the other (in any given order), this is very natural. As it will turn out, there are Nash equilibria that are not realizable as subgame perfect equilibria of such a sequential game.

As matter of fact, the analysis of subgame perfect equilibria as opposed to Nash equilibria is one of the major technical contributions of this paper. At the time of writing the conference publication~\cite{deJong2013} underlying this full-length paper, the idea of considering sequential versions of games, and Selten's subgame perfect equilibria \cite{Selten65} as an alternative to avoid the ``curse of simultaneity'' of Nash equilibria had just been brought up by Paes Leme et al.~\cite{PaesLeme}. In contrast to the price of anarchy which relates the outcome of the worst possible Nash equilibrium to that of an optimal solution \cite{Papadimitriou2001,Koutsoupias}, the \emph{sequential price of anarchy}~\cite{PaesLeme} relates the outcome of the worst possible subgame perfect equilibrium of all sequential versions of the game where players act subsequently (and farsighted), to the outcome of an optimal solution.
For set packing games, it is not hard to see (see Theorem~\ref{thm:SPE=NE} below) that any outcome of a subgame perfect equilibrium of a sequential version of the game is also a Nash equilibrium in the single-shot, strategic form of the game, but not vice versa. But this is not true in general. See, e.g.\ \cite{deJong2015} for a network routing counterexample where the sequential price of anarchy is unbounded, while the price of anarchy is known to be 5/2 \cite{AwerbuchSTOC2005,Christodoulou2005}. Indeed, subsequent to \cite{PaesLeme}, for a handful of problems it was shown that the sequential price of anarchy is lower than the price of anarchy \cite{Hassin,deJong2013,deJong2014,PaesLeme}, while for some others this is exactly opposite \cite{Angelucci,Bilo,deJong2015}. 

As mentioned earlier, our results are for $\alpha$-approximate solutions for all equilibrium concepts that we address, and any $\alpha\ge 1$. The idea to consider such relaxed notions of equilibrium also appears in early publications on the price of anarchy, such as \cite{Roughgarden}. The motivation is two-fold. First, one may argue that it is not realistic that a player $i$ willing to switch strategies for small deviations. That said, a player $i$ may be content already when 
$S_i$ is an $\alpha$-approximate best response to $S_{-i}$; see~\eqref{eq:a-eq} below.
Moreover, it is conceivable that players are bound by their computational resources, and because of that are not able to play optimally. To give a concrete example, consider the throughput scheduling example where each player $i$ owns a single machine, and the feasibility system is all sets of job $j\in J$ that can be feasibly scheduled on that machine. In the 3-field notation of \cite{graham1979optimization}, this problem reads \abc{1}{r_j}{\sum w_jU_j}, where ``1'' stands for one single machine, $r_j$ specifies that there are release dates, and the objective $\sum w_jU_j$ is to minimize the total weight jobs that finish after their duedate $d_j$ (equivalently, maximize the number of jobs scheduled before their duedate $d_j$). 
In that case, the input of the problem would realistically not be a list of all feasible sets $\mathcal{S}_i$, but the input would be the set of jobs $j\in J$ with their time windows $[r_j,d_j]$, processing times $p_j$ and values $w_j$.
It follows from Lenstra et al.~\cite{LRKB1977} that the problem to compute a best response describes an $\mathsf{NP}$-hard optimization problem. 
More generally, if players control a set of several (unrelated) machines each, the problem to compute a best response
reads \abc{R}{r_j}{\sum w_jU_j} (``R'' for unrelated machines), which is equivalent to the throughput scheduling problem as it has been addressed by Bar Noy et al. \cite{BarNoy2001}, and subsequently in \cite{BermanDasgupta2000}. For this problem, and when computation of players is bound to be polynomial time, only constant factor approximation algorithms are available. 

Two interesting special cases of throughput scheduling exist where players are able to compute an optimal play. One is when feasibility 
sets $\mathcal{S}_i$ are the sets of jobs that cane be feasibly scheduled on a single machine, and jobs have unit weights and zero release dates. This problem is solved in polynomial time by the Moore-Hodgson algorithm \cite{Moore1968}. Another is when the feasibility system $\mathcal{S}_i$ is the set of jobs with unit processing times that can be scheduled on a set of identical, parallel machines. This problem can be solved as an assignment problem~\cite{Brucker2004}.

\section{Preliminaries}\label{sec:prelimSetPacking}

We fix some notation and the basic definitions. There are $n$ players, and a finite ground set $J$ of items. Each item $j\in J$ has a value $w_j$. For $S\subseteq J$, we let $w(S):=\sum_{j\in S}w_j$.
Each player $i$ has a strategy set $\mathcal{S}_i\subseteq 2^J$ which is downward closed, i.e., if $S_i\in \mathcal{S}_i$, then $T_i\in \mathcal{S}_i$ for all $T_i\subseteq S_i$. Given a strategy profile $(S_1,\dots,S_n)$, as usual define $S_{-i}:=(S_1,\dots,S_{i-1},S_{i+1},\dots,S_n)$ as the strategies of all players except $i$, and for any set of players $K\subseteq \{1,\dots,n\}$, define $S_{-K}$ accordingly.

When $(S_1,\dots,S_n)$ is a strategy profile with $S_i\in\mathcal{S}_i$ for all $i=1,\dots,n$, the payoffs for player $i$ are defined as
\[
w(S_i,S_{-i})=
\begin{cases}
w(S_i) & \text{if } S_i\cap S_k=\emptyset \text{ for all } k\neq i\,,\\
-\infty & \text{otherwise\,.}
\end{cases}
\]
A strategy profile $(S_1,\dots,S_n)$ is an $\alpha$-approximate \emph{Nash equilibrium} (for $\alpha\ge 1$) if it is true that for all players $i=1,\dots,n$
\begin{equation}\label{eq:a-eq}
w(S_i,S_{-i})\ge \frac1\alpha w(T_i,S_{-i})\quad\text{for all}\quad T_i\in \mathcal{S}_i\,. 
\end{equation}
Note that the existence of Nash equilibria with $w(S_i,S_{-i})\ge 0$ for all players $i$ is guaranteed by the fact that the feasibility systems $\mathcal{S}_i$ are downward closed. 

\newcommand{\na}{\text{$N\!E$}}\newcommand{\opt}{\text{$O\!PT$}}
For a solution $S=(S_1,\dots,S_n)$, in a slight but convenient abuse of notation\footnote{We use $S$ to denote both, a strategy vector $S=(S_1,\dots,S_n)$ as well as the total set of items that it induces, i.e., $S=\cup_{i=1}^nS_i$. That will not yield any confusion, however.} let us write $w(S):=\sum_{i=1}^nw(S_i)$ for the total value that it achieves.
The \emph{price of anarchy} (PoA) \cite{Papadimitriou2001,Koutsoupias} for a class of games $\mathcal{I}$ is then the ratio 
\begin{equation}\label{def:poa}
\text{PoA}=\sup_{I\in \mathcal{I}}\sup_{S\in \na(I)}{\frac{w(\opt(I))}{w(S)}},
\end{equation}
where $\na(I)$ denotes the set of all $\alpha$-approximate Nash equilibria of instance $I\in\mathcal{I}$. Note that for set packing games, $\opt(I)$ is a Nash equilibrium too, hence the price of stability as proposed in \cite{Anshelevich} equals~1.

Next, consider the extensive form game that is obtained when imposing some order, say $1,\dots, n$ on the players. A strategy for player $i$ is then more complex, as it must specify one action $S_i$ for all possible combinations of actions of preceding players $1,\dots,i-1$, that is, one action $S_i$ for each node of the game tree on level~$i$.  
An $\alpha$-approximate subgame perfect equilibrium is then a strategy that guarantees at least a $1/\alpha$-fraction of the optimal action for each of the nodes of the game tree on level $i$. As we deal with a full information game, ($\alpha$-approximate) subgame perfect equilibria can be computed via backward induction\footnote{E.g., see \cite{Peters}. That is conceptually simple but generally not polynomial time.}. A nice feature of set packing games is that the computation of ($\alpha$-approximate) subgame perfect equilibria is not suffering from the typical hardness results for sequential games that is due to farsighted behaviour of players: Indeed, computing outcomes of subgame perfect equilibria may be PSPACE-hard with $n$ players \cite{PaesLeme}, and NP-hard even with two players only  \cite{deJong2015}. For set packing games, an optimal action for the $i$-th player, upon observing the actions $S_1$, \dots, $S_{i-1}$ of the preceding players, is computed by solving the optimization problem
\[
\max_{T\subseteq J} w(T)\text{ s.t.\ } T\subseteq J\setminus \cap_{k=1}^{i-1}S_k\text{ and } T\in \mathcal{S}_i\,,
\]
This suffices, as by the specific payoff structure of set packing games, the value attained by player $i$ is no longer affected by payers $i+1$,\dots, $n$ (as long as they are rational).  This problem is computationally hard only if the combinatorial structure encoded by $\mathcal{S}_i$ is hard.

The price of anarchy for $\alpha$-approximate subgame perfect equilibria, also called \emph{sequential \textup{PoA}} \cite{PaesLeme}, is then defined analogously to the price of anarchy in~\eqref{def:poa},
\newcommand{\se}{\text{$S\!P\!E$}}
\begin{equation}\label{def:spoa}
\text{sequential PoA}=\sup_{I\in \mathcal{I}}\sup_{S\in \se(I)}{\frac{w(\opt(I))}{w(S)}},
\end{equation}
where the first supremum $\sup_{I\in\mathcal{I}}$ is also taken over all possible orders of players, and $\se(I)$ denotes  
all outcomes that can be obtained as $\alpha$-approximate subgame perfect equilibria of instance~$I$.

Finally, assume that up to $k$ of the given $n$ players may collude, and are allowed to use any profit-sharing rule among them. In other words, we can think of a group $K$ of up to $k$ players as maximizing their joint value $w(S_K):=\sum_{i\in K}w(S_i)$. Then an $\alpha$-approximate $k$-collusion Nash equilibrium is a strategy profile $(S_1,\dots,S_n)$ such that the following is true for all 
sets $K$ of at most $k$ players, 
\begin{equation}\label{eq:k-eq}
w(S_K,S_{-K})\ge \frac1\alpha w(T,S_{-K})\quad\text{for all}\quad T=\cup_{i\in K}T_i\text{ and } T_i\in \mathcal{S}_i\,. 
\end{equation}
Obviously, the price of anarchy for $\alpha$-approximate $k$-collusion Nash equilibria is then again defined analogously to the price of anarchy in~\eqref{def:poa} by
\newcommand{\ce}{\text{$C\!E$}}
\begin{equation}\label{def:spoa}
\text{$k$-collusion PoA}=\sup_{I\in \mathcal{I}}\sup_{S\in \ce_k(I)}{\frac{w(\opt(I))}{w(S)}}\,,
\end{equation}
where $\ce_k(I)$ denotes the set of
$\alpha$-approximate $k$-collusion Nash equilibria of instance~$I$.

\section{An Illustrating Example}
To illustrate our definitions, consider the following, simple example.
\begin{example}\label{ex:trivial}
Assume that we have $n=2$ players and two items $J=\{1,2\}$, with weights $w_1=w_2=1$, and the feasible subsets are $\mathcal{S}_1=\{\emptyset,\{1\},\{2\}\}$ and  $\mathcal{S}_2=\{\emptyset,\{2\}\}$.
\eoe
\end{example}

Then we obviously have that $\opt=(\opt_1,\opt_2)=(\{1\},\{2\})$ is an optimal solution with $w(\opt)=2$. Next to \opt, the solution $S=(S_1,S_2)=(\{2\},\emptyset)$ is a Nash equilibrium, too, because $\{1\}\not\in\mathcal{S}_2$. That yields that this instance has \poa=2.
The strategic form of this game is depicted in Figure~\ref{fig:strat_game}.
\begin{figure}
\centering
\begin{tabular}{cc!{\vrule width 1pt}c|c|}
    \cline{3-4}
  & & \multicolumn{2}{|c|}{player 2}\\
    \cline{3-4}
   &  & $\emptyset$ & \{2\} \\
       \thickhline
       \multicolumn{1}{|c}{\multirow{3}{*}{player 1}} &
        \multicolumn{1}{|c!{\vrule width 1pt}}{$\emptyset$} & \hspace*{2.5ex}0,0\hspace*{2.5ex}& $0,1$  \\ 
        \cline{2-4}
        \multicolumn{1}{|c}{} &
        \multicolumn{1}{|c!{\vrule width 1pt}}{\{1\}} & {1,0} & $\mathbf{1,1}$ \\ 
        \cline{2-4}
        \multicolumn{1}{|c}{} &
        \multicolumn{1}{|c!{\vrule width 1pt}}{\{2\}} & $\mathbf{1,0}$ & $-\infty,-\infty$  \\ 		        
        \cline{2-3}
\hline
\end{tabular}
\caption{Strategic form for Example~\ref{ex:trivial} with Nash equilibria in bold.}\label{fig:strat_game}
\end{figure}

When considering the sequential game where player 1 precedes player 2, this yields a game tree that is depicted in Figure~\ref{fig:ext_game}. Here, all Nash equilibria are also obtained as subgame perfect equilibria, namely $(\{1\},(\{2\},\emptyset,\{2\}))$ and $(\{2\},(\{2\},\emptyset,\{2\}))$, with outcomes $(S_1,S_2)=(\{1\},\{2\})$ and $(\{2\},\emptyset)$ and 
corresponding payoffs $(1,1)$ and $(1,0)$, respectively. The worst case subgame perfect equilibrium is indicated in bold in Figure~\ref{fig:ext_game}. For the reverse order of the sequential game (player $2\to$ player 1), the only subgame perfect equilibria
are $(\{2\},(\{1\},\{1\}))$ and $(\{2\},(\{2\},\{1\}))$, with as unique outcome $(S_1,S_2)=(\{1\},\{2\})$ and corresponding payoff $(1,1)$.  As the sequential \poa\ takes the worst case over all possible sequential games, Example~\ref{ex:trivial} has sequential \poa=2. In general, Nash equilibria exist which are not sequentially realizable at all (cf.\ also \cite{Milchtaich1998}); see also Example~\ref{ex:PoAIdenticalSetPacking} below.
\begin{figure}
\centering
\includegraphics[width=0.5\textwidth]{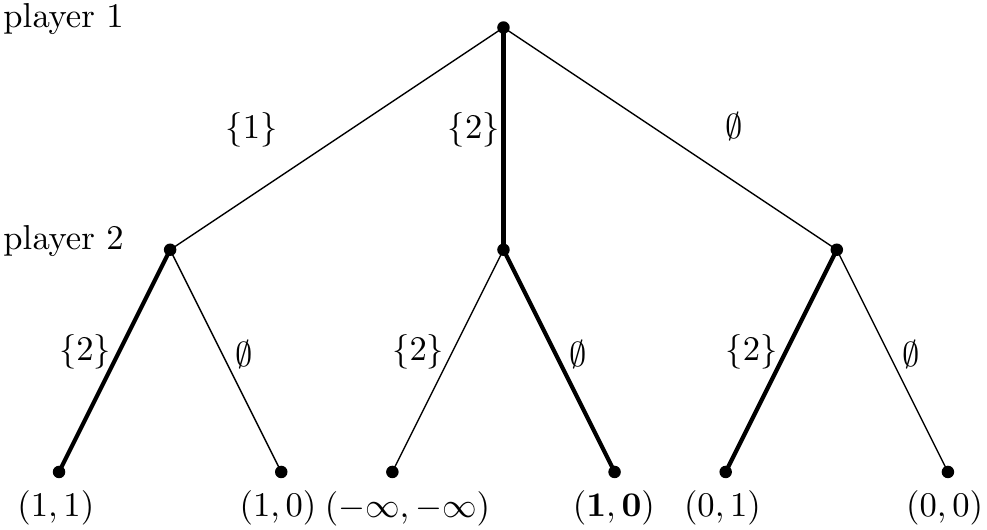}
\caption{Game tree for sequential version (player $1\to$ player 2) of Example~\ref{ex:trivial}.}\label{fig:ext_game}
\end{figure}

Finally, assume that both players collude, then obviously, the only allocation that maximizes their joint payoff is 
$(S_1,S_2)=(\{1\},\{2\})$ and corresponding payoff $(1,1)$. Therefore, the $2$-collusion $\poa=1$.

\section{Warmup: The Price of Anarchy}\label{sec:AANE}

We begin by giving the simple proof for the upper bound on the price of anarchy for arbitrary set packing games.
\begin{theorem}\label{thm:UB2}
$\poa\le \alpha+1$ for set packing games, assuming that all players play $\alpha$-approximate Nash equilibria.
\end{theorem}
\begin{proof} 
Take any instance with optimal solution \opt\ and Nash equilibrium $S$, and let $S_i$ and ${\opt}_i$, $i=1,\dots,n$, be the items selected by  player $i$ in $S$ and \opt, respectively.
For $W\subseteq J$, let $\overline{W}=J\setminus W$ be the complement of $W$ in~$J$.

Since all items in $\overline{S}$ are available, and all items in ${\opt}_i$ are feasible for player $i$, and all $\mathcal{S}_i$ are downward closed, by the definition of $\alpha$-approximate Nash equilibrium we have for all players~$i$ that
$\alpha w(S_i)\geq w(\opt_i\cap\overline{S})$.
Now we get, by using linearity of the objective function across players,
\begin{eqnarray*}
(\alpha+1) w(S) &\geq& \alpha w(S)+w(\opt\cap S) \\
&=& \sum\nolimits_{i=1}^n \alpha w(S_i)+ w(\opt\cap S)\\
&\geq& \sum\nolimits_{i=1}^n w(\opt_i \cap \overline{S})+ w(\opt\cap S)\\
&=& w(\opt)\,.
\end{eqnarray*}\qed
\end{proof}

Next we give a matching lower bound example, which is in fact a simple instance for throughput scheduling (yet an asymmetric set packing game, see Section~\ref{sec:seq} below).
\begin{example}\label{ex:related2} Assume without loss of generality that $\alpha={p}/{q}$, where $p\geq q$. Consider a game with $q+1$ players. For each player $i$, there is one machine, which we also denote by $i$. The set $J$ of items are jobs that are partitioned into two sets $P$ and $Q$, with $|P|=p,|Q|=q$. Each job $j\in J$ has deadline $d_j=1$, unit weight $w_j=1$, and its processing time on machine 1 is $p_{j1}=1/p$. Moreover, jobs $j\in Q$ have processing time $p_{ji}=1$ on any other machine $i\neq 1$, while jobs $j\in P$ have processing time $p_{ji}=2$ on any other machine $i\neq 1$. Note that any subset of jobs of size $p$ can be feasibly allocated to player 1. Players $2\dots n$ can be allocated only one job each, and only jobs from $Q$. See Figure~\ref{vb2} for an illustration in the case where $\alpha=3/2$. \eoe
\end{example}
\begin{figure}[tbhp]
   \centering
   \includegraphics[width=0.6\textwidth]{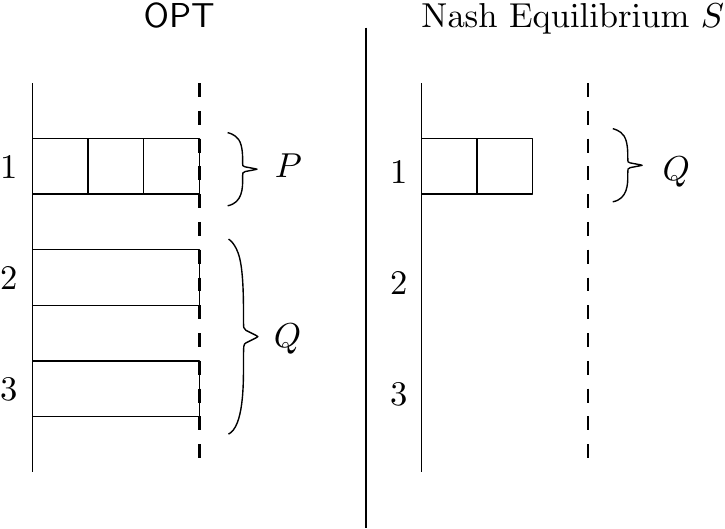}\caption{Example~\ref{ex:related2} for $p=3$ and $q=2$. Numbers represent machines. Rectangles represent jobs. The left side of each job is its starting time, its width is its processing time on the machine on which it is allocated. The dashed line is the deadline, which is the same for all jobs in this example.}\label{vb2}
   \end{figure}

\begin{theorem}\label{lem:LBGeneralSetPacking}
$\poa\geq\alpha+1$ for throughput scheduling games (and hence also for set packing games), assuming that all players play $\alpha$-approximate Nash equilibria.
\end{theorem}
\begin{proof}
In the  optimum solution $\opt$, all $p+q$ jobs are feasibly allocated: All jobs in $P$ are allocated to player 1, and each of the jobs in $Q$ is allocated to one of the $q$ other players $2,\dots, q+1$.
Now consider the $\alpha$-approximate Nash equilibrium $S$ where only $q$ jobs are allocated: All jobs from $Q$ are allocated to player 1, and no jobs are allocated to players $2, \dots, q+1$. This is indeed an $\alpha$-approximate Nash equilibrium, as player 1 
achieves a total value of $q$, while maximally that player can be allocated $p=\alpha q$ jobs. In other words, the $\alpha$-approximate Nash condition \eqref{eq:a-eq} holds for player 1. Moreover, given that all jobs from $Q$ are allocated to player $1$, players $2, \dots, q+1$ cannot do better than a value 0, as none of the jobs from $P$ are feasible for these players. We conclude that $\poa\geq w(\opt)/w(S)= (p+q)/q=\alpha+1$.\qed
\end{proof}

Note that when $\alpha$ is not rational, we can obtain a price of anarchy arbitrarily close to $\alpha+1$ by letting $p/q$ approach $\alpha$. Also recall that $\alpha=1$ for the special case where the players can verify whether a solution is a Nash equilibrium, which yields the following.
 \begin{corollary}
$\poa=2$ for set packing games and throughput scheduling games.
\end{corollary}

Finally note that the upper bound is universal in the sense that it is independent of how the ($\alpha$-approximate) Nash equilibrium is obtained. It is conceivable that specific algorithms can yield a better bound for the price of anarchy. However, the existence of more complicated counter-examples for specific algorithms is not unlikely either (see next Section~\ref{sec:seq} for an example).

\section{The Sequential Price of Anarchy}\label{sec:seq}
It is actually not difficult to see that the example that we have used in the preceding section as a lower bound example for Nash equilibria, also provides a lower bound for subgame perfect equilibria. Hence we get the following for free.
\begin{theorem}\label{thm:SPoAUB2}
The sequential $\poa=\alpha+1$ for set packing games and throughput scheduling games, assuming that players play $\alpha$-approximate subgame perfect equilibria.
\end{theorem}
\begin{proof}
Recall Example~\ref{ex:related2}, and assume that player 1 is the first to make a selection. Then if player 1 makes the same selection of job set $Q$ as in the proof of Theorem~\ref{lem:LBGeneralSetPacking}, the obtained solution can indeed be obtained as an $\alpha$-approximate subgame perfect equilibrium, as player 1 cannot improve by more than a factor $\alpha$ by selecting other jobs, and given that, all other players have nothing to choose. (We can specify any reasonable actions for those parts of the game tree that are not played in this equilibrium.)
By the same argument as before, the lower bound on the price of anarchy follows. \qed
\end{proof}

To finish the proof of Theorem~\ref{thm:SPoAUB2}, observe that the upper bound of Theorem~\ref{thm:UB2} also carries over, by the subsequent theorem. 
\begin{theorem}\label{thm:SPE=NE}
For set packing games, the actions played in an $\alpha$-approximate subgame perfect equilibrium of any sequential version of set packing game define an $\alpha$-approximate Nash equilibrium in the original, single-shot game.
\end{theorem}

\begin{proof}
Consider the actions $S=(S_1,\dots,S_n)$ played in any subgame perfect equilibrium $\se$ of any sequential version of the set packing game. Assume w.l.o.g.\ the order was $1, \dots, n$. Consider any player $i$ choosing $S_i$. As the choice $S_i$ is part of a subgame perfect strategy, we know
$\alpha w(S_i) \geq  w(T_i)$  for all  $T_i\in \mathcal{S}_i$  with  $T_i\subseteq J\setminus\cup_{k=1}^{i-1}S_k$, since in a subgame perfect equilibrium, $i$'s payoff is not 
affected by (rational) subsequent players $k>i$, for any such $T_i$.  This 
because, for any such $T_i$,
none of the subsequent players $k>i$ will choose to select an $S_k$ with $T_i\cap S_k\neq\emptyset$.
But this of course also implies that 
$\alpha w(S_i) \geq  w(T_i)$  for all  $T_i\in \mathcal{S}_i$  with  $T_i\subseteq J\setminus\cup_{k\neq i}S_k$,
by the same argument. This is exactly the Nash condition~\eqref{eq:a-eq}, which is true for all players $i$. Hence $S$ a Nash equilibrium in the original, single-shot game.
\qed
\end{proof}
Note that this is not true in general. See, e.g., \cite{deJong2015} for an example. Indeed, it is a result of the definition of payoffs for set packing games. Finally, for $\alpha=1$, we obtain the following.
 \begin{corollary}
The sequential $\poa=2$ for set packing games and throughput scheduling games.
\end{corollary}

\section{Symmetric Set Packing Games}\label{sec:IM}
We call a set packing game \emph{symmetric} whenever there is only one feasibility system $\mathcal{S}$ that is the same for all players $i$, but a player $i$ can select $x_i$ feasible sets from $\mathcal{S}$, for some integer $x_i\geq 1$.
Note that when all $x_i=1$, this exactly means that the strategic form game is symmetric in the sense that all players have exactly the same strategy set. However we choose to 
allow players to select multiple feasible sets. In the throughput scheduling context, that would be a player who controls $x_i$ identical machines. We define $x:=\sum_{i=1}^nx_i$ to be the total set of feasible sets from $\mathcal{S}$ that can be chosen by all players together, and note that $x\geq n$.

In this section we show that the symmetric version of set packing games allows an improvement in the price of anarchy when considering sequential games and subgame perfect equilibria. In the light of Theorem~\ref{thm:SPE=NE}, that boils down to the statement that some of the Nash equilibria that are responsible for the price of anarchy of $\alpha+1$ (which also holds for symmetric set packing games, see Theorem~\ref{thm:PoAIdenticalSetPacking} below), are not achievable by sequential play, hence they are probably not realistic.

\subsection*{The Price of Anarchy}

\begin{theorem}\label{thm:PoAIdenticalSetPacking}
$\poa=\alpha+1$ for symmetric set packing games, assuming that all players play $\alpha$-approximate equilibria.
\end{theorem}

The upper bound $\alpha+1$ is a consequence of Theorem~\ref{thm:UB2}. The lower bound follows from the following example, which is again an example where the feasibility sets are defined by a throughput scheduling problem. Symmetry means that all machines are identical.
\begin{example}\label{ex:PoAIdenticalSetPacking} Let $\alpha=p/q$. There are $n$ players~$i$, each corresponding to one machine. The set $J$ of $p+(q+1)(n-1)$ jobs is again partitioned into two sets $P,Q$, $|Q|=q(n-1)+p, |P|=(n-1)$. All jobs $j\in J$ have deadline $d_r=1$. Job $j\in Q$ have processing times $p_j=1/(q(n-1)+p)$ and weight $w_j=1$, while jobs $j\in P$ have processing times $p_j=1$ and weight $w_j=p$. See Figure~\ref{vb3} for an illustration for the case where $p=3,q=2$ and $n=3$ \eoe
\end{example}
\begin{figure}[tbhp]
   \centering
   \includegraphics[width=0.6\textwidth]{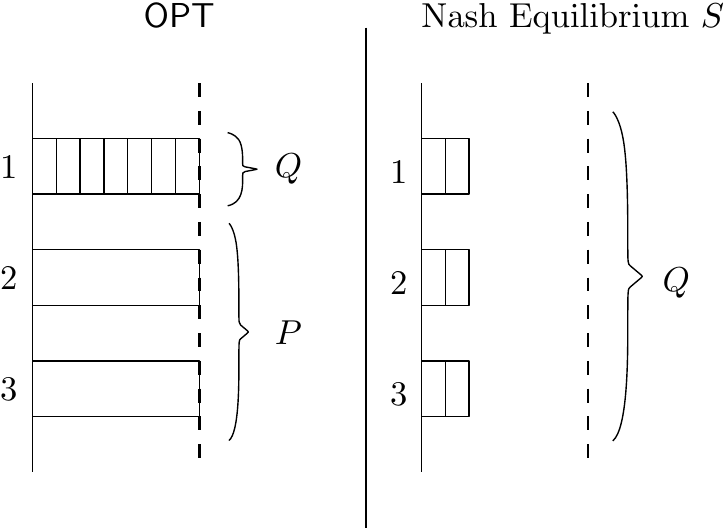}\caption{Example~\ref{ex:PoAIdenticalSetPacking} for $p=3,q=2,n=3$. Numbers represent machines. Rectangles represent jobs. The left side of each job is its starting time, its width is its processing time. The dashed line is the deadline, which is the same for all jobs in this example.}\label{vb3}
\end{figure}
\begin{proof}[of the lower bound]
In the optimum solution $\opt$, player $1$ is allocated all jobs in $Q$, and each other player is allocated exactly one job in $P$. Consider Nash equilibrium $S$ where each player is allocated $q$ jobs in $Q$. Note that $S$ is indeed an $\alpha$-approximate Nash equilibrium: Any player $i$ could choose at most one job from $P$ or at most $p$ jobs from $Q$, since other players are allocated $q(n-1)$ jobs from $Q$ in total. Neither of the feasible deviations increases player $i$'s utility by more than a factor $\alpha$. For this example, $w(\opt)/w(S)=\frac{pn+q(n-1)}{qn}=\frac{p+q}{q}-\frac{1}{n}\to 1+\alpha$ for $n\to\infty$.\qed
\end{proof}
For $\alpha=1$, we obtain the following.
 \begin{corollary}
$\poa=2$ for symmetric set packing games and throughput scheduling games with identical machines.
\end{corollary}
Note that (for $\alpha=1$) this Nash equilibrium is not subgame perfect in the corresponding sequential game; in any subgame perfect equilibrium, the first player would necessarily choose all jobs from $Q$.

\subsection{Sequential Price of Anarchy}\label{LB3}
In contrast to the asymmetric case, subgame perfect equilibria indeed rule out some of the bad quality Nash equilibria when considering symmetric set packing games. The main result of this section is:
\begin{theorem}\label{thm:ident1}
$\text{The sequential \poa}={\sqrt[\alpha]{e}}/({\sqrt[\alpha]{e}-1})$ for symmetric set packing games, when all players play $\alpha$-approximate subgame perfect equilibria.
\end{theorem}

First we prove the lower bound, which is again a throughput scheduling instance.
\begin{example}\label{ex:SPidenticalSPoA}
There are $n$ players. Each player $i$ corresponds to one machine. The set $J$ of $n^2$ jobs is partitioned into $n$ sets $J_1,\dots,J_n$, $|J_k|=n$ for all $k\in [n]$. We refer to a job from $J_k$ as a $k$-job. All $k$-jobs have deadline $k$. All jobs $j\in J$ have processing time $p_j=1$ and weight $w_j=1$. See Figure~\ref{vb4} for an illustration for the case where $n=5$ and $\alpha=1$. \eoe
\end{example}
\begin{lemma}\label{thm:LBident1}
The sequential $\text{\poa}\geq{\sqrt[\alpha]{e}}/({\sqrt[\alpha]{e}-1})$ for identical set packing games, when all players play $\alpha$-approximate subgame perfect equilibria.
\end{lemma}
\begin{proof}

\begin{figure}[tbhp]
   \centering
    \includegraphics[width=0.8\textwidth]{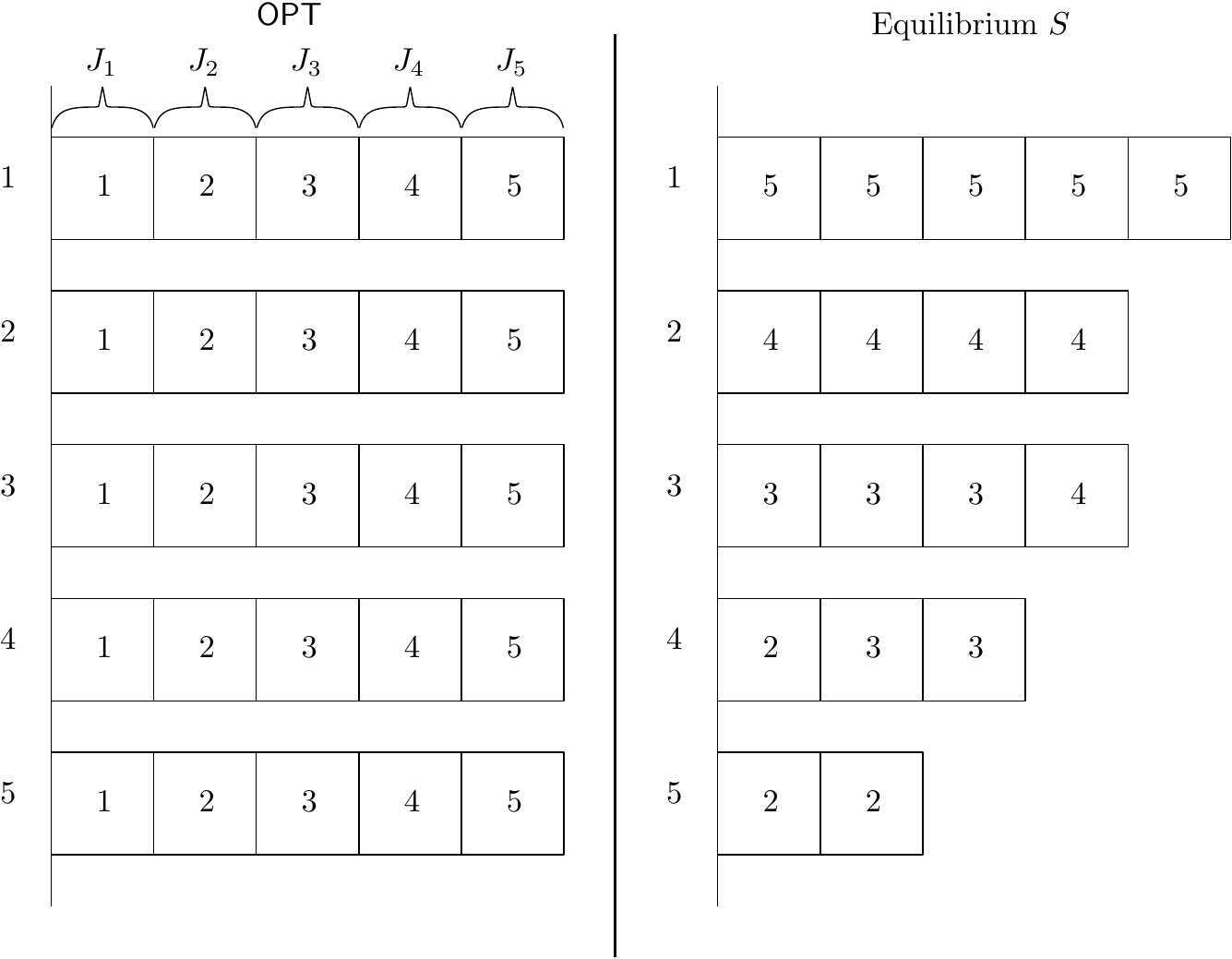}\caption{Example~\ref{ex:SPidenticalSPoA} in case of 5 players and $\alpha=1$. Numbers represent machines. Rectangles represent jobs. The left side of each job is its starting time, its width is its processing time. The number in each job is its deadline.}\label{vb4}

\end{figure}

In the optimum solution $\opt$, every player is allocated exactly one $k$-job for all $k=1,\dots,n$. Therefore $w(\opt)=n^2$.

We construct an $\alpha$-approximate subgame perfect outcome $S$, as follows: For every player $i=1,\dots,n$ in this order, we find the maximum number of jobs that can be feasibly allocated to this player, given jobs already assigned to players $1,\dots,i-1$, and when considering jobs with the largest deadlines first (which are the most flexible jobs). Denote this number of jobs $m_i$.
We allocate to player $i$ exactly $\lceil{m_i}/{\alpha}\rceil$ of these jobs, so that the allocation is still an $\alpha$-approximation. Let $S_i$ be the jobs allocated to player $i$ in this way.

We bound $w(S)$ in the following way: Let $r_k(i)=\frac{|S_i\cap J_{k}|}{|S_i|}$, i.e.\ $r_k(i)$ is the fraction of $k$-jobs allocated to player $i$, relative to the total number of jobs allocated to player $i$. Let $r_k=\sum_{i=1}^n r_k(i)$.
Now,
\begin{equation}\label{eq:sumdelta}
\sum_{k=1}^n r_k=\sum_{k=1}^n \sum_{i=1^n} r_k(i)=\sum_{i=1}^n \sum_{k=1}^n \frac{|S_i\cap J_k|}{|S_i|}=\sum_{i=1}^n1=n\,.
\end{equation}
In $S$, any player $i$ who gets allocated a $k$-job, is not allocated any job from $J_{j},j\geq k+2$, hence she is allocated at most $\lceil{(k+1)}/{\alpha}\rceil \leq {(k+1+\alpha)}/{\alpha}$ jobs. Therefore, each $k$-job contributes at least ${\alpha}/{(k+1+\alpha)}$ to $r_k$. For any $k$ for which all of the $n$ $k$-jobs are allocated in $S$, we obtain
\begin{equation}\label{eq:rdeltaLB}
r_k \geq {n\alpha}/{(k+1+\alpha)}\,.
\end{equation}

Now, for some $k'\geq 0$, by construction of the allocation we have that for all $k\geq n-k'$, all $n$ $k$-jobs are allocated, as well as a subset of the $(n-(k'+1))$-jobs. We obtain

\begin{equation}\label{proof:id}
n\geq\sum_{k=n-k'}^{n} r_k \geq \sum_{k=n-k'}^{n} \frac{n\alpha}{k+1+\alpha} \geq \int_{k=n-k'}^{n}\frac{n\alpha}{k+1+\alpha}dk\,\,,
\end{equation}
where the first inequality follows from \eqref{eq:sumdelta}, the second inequality follows from \eqref{eq:rdeltaLB}, and the last inequality follows from basic calculus.

Because the last term is upper bounded by $n$, we can derive an upper bound on~$k'$. In fact, basic calculus shows that \[
k'>\frac{(n+1+\alpha)(\sqrt[\alpha]{e}-1)}{\sqrt[\alpha]{e}}\ \Rightarrow\ \int_{k=n-k'}^{n}\frac{n\alpha}{k+1+\alpha}dk > n\,, 
\] 
which together with \eqref{proof:id} yields that
$k'\leq\frac{(n+1+\alpha)(\sqrt[\alpha]{e}-1)}{\sqrt[\alpha]{e}}$. Because only $k$-jobs with $k\geq n-(k'+1)$ are allocated, we conclude that 
\[ 
w(S)\leq (k' +1)n\leq
\frac{(n+1+\alpha+\frac{\sqrt[\alpha]{e}}{\sqrt[\alpha]{e}-1})(\sqrt[\alpha]{e}-1)}{\sqrt[\alpha]{e}}\cdot n\,. \] We see that \[ \frac{w(\opt)}{w(S)} \geq
\frac{n\sqrt[\alpha]{e}}{(n+1+\alpha+\frac{\sqrt[\alpha]{e}}{\sqrt[\alpha]{e}-1})(\sqrt[\alpha]{e}-1)}\to\frac{\sqrt[\alpha]{e}}{\sqrt[\alpha]{e}-1}\quad\text{for }\ n\to\infty\,, \] and the claim
follows.\qed
\end{proof}

Note that the lower bound construction assumes that players choose the most flexible jobs first, which seems reasonable from a practical point of view. Also note that in the lower bound example, $x_i=1$ for all players. Therefore, the lower bound holds even in the special case when the strategic form game is symmetric.

To derive a matching upper bound on the sequential price of anarchy for symmetric set packing games, we use a proof idea from Bar-Noy~et al.~\cite{BarNoy2001} in their analysis of $k$-GREEDY, but we generalize it for the case where $x_i> 1$ for some players $i$  (e.g., player $i$ controls multiple identical machines in the throughput scheduling setting).

We want to prove:
\begin{theorem}\label{thm:UBident1}
The sequential $\text{\poa}\leq{\sqrt[\alpha]{e}}/({\sqrt[\alpha]{e}-1})$ for symmeyric set packing games, when all players play $\alpha$-approximate subgame perfect equilibria.
\end{theorem}

Denote by $S_i$ the items selected by player $i$ in an $\alpha$-approximate subgame perfect equilibrium, and recall that $S$ denotes both the strategy vector and $S=\cup_{i=1}^n S_i$, the total set of selected items. The following lemma lower bounds the total weight collected by player~$i$.
\begin{lemma}\label{lem:WSI_}
We have for all players $i$
\[w(S_i)\geq \frac{x_i}{x\alpha} w\left(\opt\left(J\setminus \cup_{j<i}S_j\right)\right)\,.\]
where  $\opt(W)$ denotes an optimal solution for any subset of items $W\subseteq J$.
\end{lemma}
\begin{proof}
Let $W:=J\setminus \cup_{j<i}S_j$.
Let $\opt^i$ denote the maximum weight set of items that player $i$ can achieve from $W$. Observe that $w(\opt^i)\geq (x_i/x) w(\opt(W))$.
This follows because player $i$ could potentially select the $x_i$ most valuable feasible sets from $\opt(W)$.   Now, because we assume an $\alpha$-approximate subgame perfect equilibrium, 
$w(S_i)\geq {w(\opt^i)}/{\alpha}\geq {x_i w(\opt(W))}/{(x\alpha)}$.
\qed
\end{proof}

\begin{proof}[of Theorem~\ref{thm:UBident1}]
Let $\gamma:=x\alpha$, and recall that $w(\opt)=w(\opt(J))$ denotes the value of an optimal solution. 
We use Lemma~\ref{lem:WSI_}, to get
\begin{equation*}
w(S_i) \geq  \frac{x_i}{\gamma}w\left(\opt\left(J\setminus \bigcup\nolimits_{j<i}S_j\right)\right)
\geq  \frac{x_i}{\gamma}\left(w(\opt) - \sum\nolimits_{j<i}w(S_j)\right)\,,
\end{equation*}
where the latter inequality holds because $w(\opt) - \sum\nolimits_{j<i}w(S_j)$ represents the value of a feasible solution for the items $J\setminus \bigcup\nolimits_{j<i}S_j$.
Add $\sum_{j=1}^{i-1}w(S_j)$ to both sides to get
\begin{equation}
\sum_{j=1}^{i}w(S_j)\geq \frac{x_i w(\opt)}{\gamma}+\frac{\gamma-x_i}{\gamma}\sum_{j=1}^{i-1}w(S_j)\label{e1_}\,.
\end{equation}
We want to prove by induction on $i$ that
\begin{equation}\label{eq:double_ind}
\sum_{j=1}^{i}w(S_j)\geq \frac{\gamma^{x'_i}-(\gamma-1)^{x'_i}}{\gamma^{x'_i}}w(\opt)\,,
\end{equation}
where $x'_i=\sum_{j=1}^i x_j$. 

The base case $i=1$ is the following lemma, proved by yet another inductive argument on $x_1$.
\begin{lemma}\label{lem:app1}
\[w(S_1)\geq \frac{\gamma^{x_1}-(\gamma-1)^{x_1}}{\gamma^{x_1}}w(\opt)\,.\]\label{AL1}
\end{lemma}
\begin{proof} We know by definition of $\gamma$, and by plugging $i=1$ into Lemma~\ref{lem:WSI_} that
\begin{equation*}
{w(S_1)} \geq \frac{x_1}{\gamma} {w(\opt)}\,.
\end{equation*}
Hence we are done when we can prove by induction on $x_1$ that
\begin{equation*}
\frac{x_1}{\gamma}\geq \frac{\gamma^{x_1}-(\gamma-1)^{x_1}}{\gamma^{x_1}}\,.
\end{equation*}
When $x_1=1$, we get
\begin{equation*}
\frac{1}{\gamma}\geq \frac{\gamma-(\gamma-1)}{\gamma}=\frac{1}{\gamma}\,,
\end{equation*}
which clearly holds. Assume the claim holds for $x_1=k-1$. We get
\begin{align*}
\frac{k}{\gamma}&=
\frac{k-1}{\gamma}+\frac{1}{\gamma}\\
&\geq\frac{\gamma^{k-1}-(\gamma-1)^{k-1}}{\gamma^{k-1}}+\frac{1}{\gamma}\\
&=\frac{\gamma^{k}-(\gamma-1)^{k}-{(\gamma-1)}^{k-1}+{\gamma}^{k-1}}{\gamma^{k}}\\
&\geq \frac{\gamma^{k}-(\gamma-1)^{k}}{\gamma^{k}}\,,
\end{align*}
proving Lemma~\ref{lem:app1}.\qed
\end{proof}
%
%
Assume now that \eqref{eq:double_ind} holds for $i-1$. Applying the induction hypothesis to $(\ref{e1_})$ we get
\begin{equation*}
\sum_{j=1}^{i}w(S_{j})\geq \frac{x_i w(\opt)}{\gamma}+\frac{\gamma-x_i}{\gamma}\cdot\frac{\gamma^{x'_{i-1}}-(\gamma-1)^{x'_{i-1}}}{\gamma^{x'_{i-1}}}w(\opt)\,.
\end{equation*}
This can be used to prove the inductive claim, using the following.
\begin{lemma}\label{lem:app2}
\[
\frac{x_k}{\gamma}+\frac{\gamma-x_k}{\gamma}\cdot\frac{\gamma^{x'_{k-1}}-(\gamma-1)^{x'_{k-1}}}{\gamma^{x'_{k-1}}}\geq\frac{\gamma^{x'_{k}}-(\gamma-1)^{x'_{k}}}{\gamma^{x'_{k}}}\,.
\]
\end{lemma}
\begin{proof}We have
\begin{align*}
&\frac{x_k }{\gamma}+\frac{\gamma-x_k}{\gamma}\cdot\frac{\gamma^{x'_{k-1}}-(\gamma-1)^{x'_{k-1}}}{\gamma^{x'_{k-1}}}\\
=&\frac{x_k}{\gamma}\cdot\frac{(\gamma-1)^{x'_{k-1}}}{\gamma^{x'_{k-1}}}+\frac{\gamma^{x'_{k-1}}-(\gamma-1)^{x'_{k-1}}}{\gamma^{x'_{k-1}}}\\
\geq&\frac{\gamma^{x_k}-(\gamma-1)^{x_k}}{\gamma^{x_k}}\cdot\frac{(\gamma-1)^{x'_{k-1}}}{\gamma^{x'_{k-1}}}+\frac{\gamma^{x'_{k-1}}-(\gamma-1)^{x'_{k-1}}}{\gamma^{x'_{k-1}}}\\
=&\left(1-\frac{(\gamma-1)^{x_k}}{\gamma^{x_k}}\right) \cdot\frac{(\gamma-1)^{x'_{k-1}}}{\gamma^{x'_{k-1}}}+1-\frac{(\gamma-1)^{x'_{k-1}}}{\gamma^{x'_{k-1}}}\\
=&1-\frac{(\gamma-1)^{x_k}}{\gamma^{x_k}}\cdot\frac{(\gamma-1)^{x'_{k-1}}}{\gamma^{x'_{k-1}}}\\
=&\frac{\gamma^{x'_{k}}-(\gamma-1)^{x'_{k}}}{\gamma^{x'_{k}}}\, ,
\end{align*}
where the first inequality follows from $\frac{x_k}{\gamma}\geq \frac{\gamma^{x_k}-(\gamma-1)^{x_k}}{\gamma^{x_k}}$, as shown in the proof of Lemma~\ref{AL1}, and the last equality follows from  $x'_{k}=x'_{k-1}+x_k$.\qed
\end{proof}

Hence we get for $i=n$ (see also \cite[Thm 3.3]{BarNoy2001})
\begin{equation*}
w(S)=\sum_{j=1}^nw(S_j)\geq \frac{\gamma^{x}-(\gamma-1)^{x}}{\gamma^{x}}w(\opt)\,.
\end{equation*}
We therefore get that the
\begin{equation*}
\text{sequential PoA} \leq \frac{\gamma^x}{\gamma^x-(\gamma-1)^x} = \frac{(x\alpha)^x}{(x\alpha)^x-(x\alpha-1)^x}\leq \frac{\sqrt[\alpha]{e}}{\sqrt[\alpha]{e}-1}\,,\label{IUB}
\end{equation*}
where the last inequality follows because the right hand side is exactly the limit for $x\to\infty$, and the series $b_x={(x\alpha)^x}/({(x\alpha)^x-(x\alpha-1)^x})$ is monotone in $x$, with $b_1=\alpha \leq {\sqrt[\alpha]{e}}/({\sqrt[\alpha]{e}-1})$.
This ends the proof of Theorem~\ref{thm:UBident1}\qed
\end{proof}

Basic calculus shows that
\[\alpha+\frac{1}{2}\leq \frac{\sqrt[\alpha]{e}}{\sqrt[\alpha]{e}-1} \leq \alpha+\frac{1}{e-1}\]
for $\alpha\geq 1$. Hence the improvement over the (Nash equilibrium) price of anarchy which was $\alpha+1$ is substantial. Note that for $\alpha=1, \text{PoA}={{e}}/({{e}-1})\approx1.58$.

\begin{corollary}
The sequential $\poa={{e}}/({{e}-1})\approx1.58$ for symmetric set packing games, when all players play subgame perfect.
\end{corollary}

\section{$k$-Collusion Price of Anarchy}\label{sec:dealproof}

While sequential play was a way to reduce the price of anarchy for symmetric set packing games, we now show that collusion of players helps to reduce the price of anarchy, too. This is true also for general, asymmetric set packing games. Recall that an $\alpha$-approximate $k$-collusion Nash equilibrium means that no coalition $K$ of up to $k$ players can improve their total value $w(K)$ by more than a factor $\alpha$.
\begin{theorem}
The $k$-collusion $\poa=\alpha+\frac{n-k}{n-1}$ for set packing games, when all players play $\alpha$-approximate $k$-collusion Nash  equilibria.
\end{theorem}
Note that for $k=1$, we consider $\alpha$-approximate Nash equilibria and the $1$-collusion $\poa=\alpha+1$, which is consistent with Theorem~\ref{thm:UB2}. Also note that for $k=n$, we consider an $\alpha$-approximate (centralized) solution, so for $\alpha=1$ this is just an optimal solution.

\begin{proof}
First we give an upper bound proof.
\begin{lemma}
The $k$-collusion $\poa\leq \alpha+\frac{n-k}{n-1}$ for set packing games, when the players play an $\alpha$-approximate $k$-collusion Nash equilibrium.
\end{lemma}
\begin{proof}
The proof mimics our earlier proof of Theorem~\ref{thm:UB2}, only here we have to keep track of the values of more subsets of $J$.
We fix an optimal solution $\opt$ and a $k$-collusion Nash equilibrium $S$, write $N=\{1,\dots,n\}$, and use the following notation:
\[
x_{ij}=\begin{cases}
\text{ the total weight of items  in $\opt_i\cap S_j$} & \text{for }i,j\in N\,, \\
\text{ the total weight of items  in $S_j\setminus\opt$} & \text{for }i=0,j\in N\,, \\
\text{ the total weight of items  in $\opt_i\setminus S$} & \text{for }i\in N,j=0\,.
\end{cases}
\]

Our proof is based on the following observation: Players from any coalition $K$ collude and collectively deviate if and only if the total weight of items allocated to them increases by more than a factor $\alpha\geq 1$, by choosing any set of items in $(\cup_{i\in K}S_i)\cup (J\setminus\cup_{i\not\in K}S_i)$. Therefore, in particular for all coalitions $K$ of size $k$ in any $\alpha$-approximate $k$-collusion Nash equilibrium, we have
\begin{align*}
\alpha\left(\sum_{j\in K}\left(\sum_{i\in N} x_{ij}+ x_{0j}\right)\right)\geq\sum_{i\in K} \left(\sum_{j\in K}x_{ij} + x_{i0}\right)\,.
\end{align*}
Note that all items that contribute to the left-hand side are allocated to players in $K$ in the equilibrium $S$. Also note that all items that contribute to the right-hand side can be feasibly allocated to players in $K$, since these items are allocated to players from $K$ in $\opt$. Also, these items are available for coalition $K$, since they are either allocated to players in $K$ in $S$, or not allocated. We rewrite this as
\begin{align}\label{eq:dpoacond}
\alpha\left(\sum_{j\in K}\left(\sum_{i\in N} x_{ij}+ x_{0j}\right)\right)\geq\mathop{\sum_{i\in K}\sum_{j\in K}}_{i\neq j}x_{ij} + \sum_{i\in K}\left( x_{ii} + x_{i0}\right) \,.
\end{align}
Now, any player $i$ is in ${\binom{n-1}{k-1}}$ coalitions of size $k$, and any combination of two players $i,j$ is in $\binom{n-2}{k-2}$ coalitions of size $k$. Therefore, summing \eqref{eq:dpoacond} over all coalitions $K$ of size $k$ yields
\begin{align*}
&\alpha\binom{n-1}{k-1} \left(\sum_{j\in N}\left(\sum_{i\in N} x_{ij}+ x_{0j}\right)\right)\\
\geq &\binom{n-2}{k-2}\mathop{\sum_{i\in N}\sum_{j\in N}}_{i\neq j} x_{ij}+\binom{n-1}{k-1}\sum_{i\in N}(x_{ii}+ x_{i0})\,.
\end{align*}
Adding
$$\left({\binom{n-1}{k-1} }-{\binom{n-2}{k-2} }\right)\mathop{\sum_{i\in N}\sum_{j\in N}}_{i\neq j} x_{ij}={\binom{n-2}{k-1} }\mathop{\sum_{i\in N}\sum_{j\in N}}_{i\neq j} x_{ij}$$
 to both sides yields
\begin{align}\label{eq:compdpe}
\nonumber &\alpha\left({\binom{n-1}{k-1}}+{\binom{n-2}{k-1}}\right) \mathop{\sum_{i\in N}\sum_{j\in N}}_{i\neq j} x_{ij}+ \binom{n-1}{k-1} \left(\sum_{j\in N}(x_{jj}+ x_{0j})\right)\\
\geq &\binom{n-1}{k-1}\mathop{\sum_{i\in N}\sum_{j\in N}}_{i\neq j} x_{ij}+\binom{n-1}{k-1}\sum_{i\in N}(x_{ii}+ x_{i0})\,.
\end{align}
Therefore,
\begin{align*}\label{eq:DPPOA}
 & \alpha \left( \binom{n-1}{k-1}+\binom{n-2}{k-1}\right)w(S)\\
=\ & \alpha\left( \left( \binom{n-1}{k-1}+\binom{n-2}{k-1}\right) \sum_{i\in N}\sum_{j\in N} x_{ij}+\left(\binom{n-1}{k-1}+\binom{n-2}{k-1}\right) \sum_{j\in N} x_{0j}\right)\\
\geq\ & \alpha\left({\binom{n-1}{k-1}}+{\binom{n-2}{k-1}}\right) \mathop{\sum_{i\in N}\sum_{j\in N}}_{i\neq j} x_{ij}+ \binom{n-1}{k-1} \left(\sum_{j\in N}(x_{jj}+ x_{0j})\right)\\
\geq\ & \binom{n-1}{k-1}\mathop{\sum_{i\in N}\sum_{j\in N}}_{i\neq j} x_{ij}+\binom{n-1}{k-1}\sum_{i\in N}(x_{ii}+ x_{i0})\\
=\ & \binom{n-1}{k-1}w(\opt)\,,
\end{align*}
where the last inequality follows from \eqref{eq:compdpe}. This yields
\begin{align*}
k\text{-collusion } \poa\leq \alpha \frac{{\binom{n-1}{k-1}}+{\binom{n-2}{k-1}}}{{\binom{n-1}{k-1}}}=\alpha+\frac{n-k}{n-1}\,.
\end{align*}\qed
\end{proof}

In fact, this proof of the upper bound provides us with an easy way to create a tight lower bound example for any $n$.
\begin{example}\label{ex:DPE}
We make the upper bound analysis tight by setting $x_{ii}=0$ and $x_{0i}=0$ for all players $i\in N$ . We normalize $x_{ij}=1$ for all players $i,j\in N$ for which $i\neq j$, and finally we set $x_{i0}=n-k+(n-1)(\alpha-1)$ for all players $i\in N$. We construct the strategy spaces such that any player $i$ can only choose subsets of either $\opt_i$ or $S_i$, where $S_i$ is the set chosen in the in the $k$-collusion Nash equilibrium. The resulting game for $n=3, k=2$ is shown in Figure~\ref{vb8}.\eoe
\end{example}

\begin{figure}[tbhp]
   \centering
   \includegraphics[width=\textwidth]{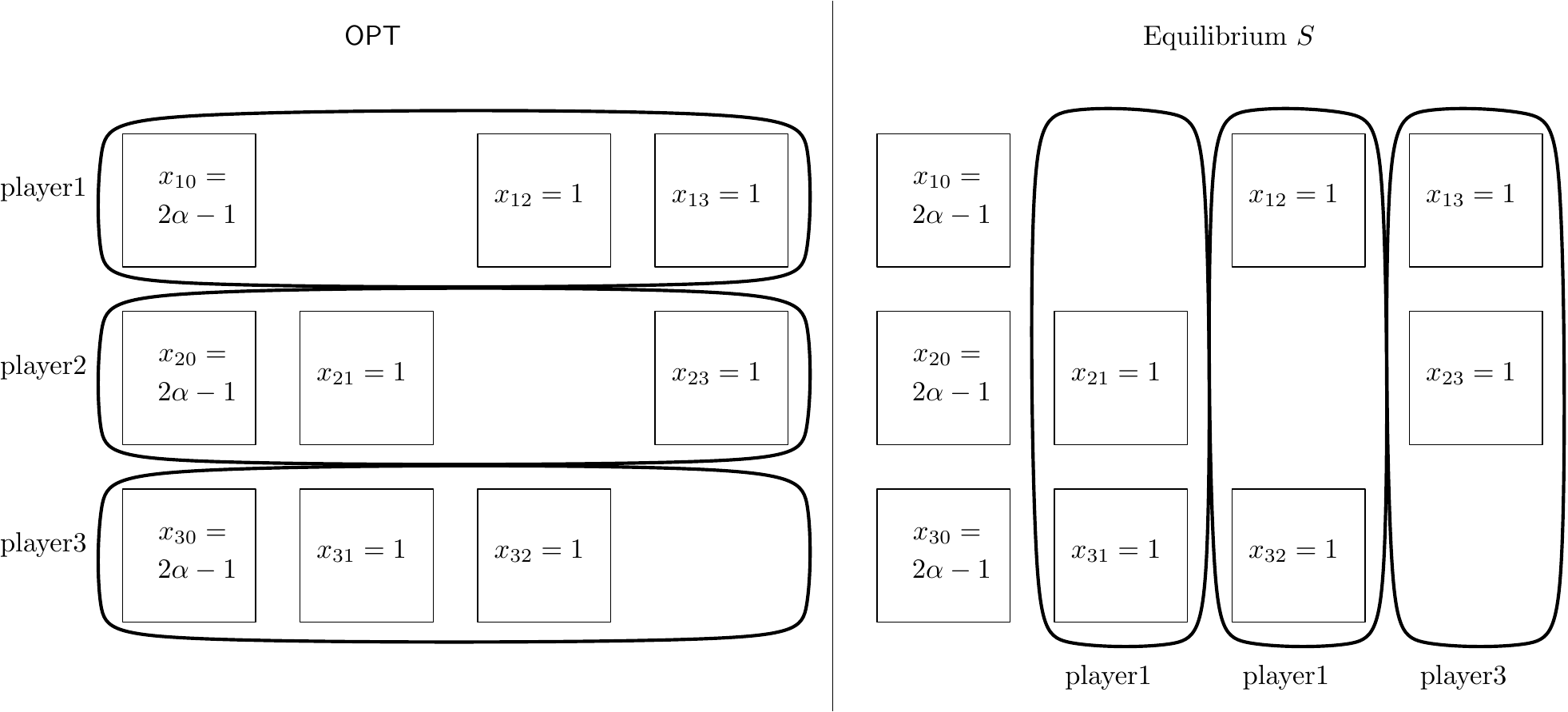}\caption{The $k$-collusion Nash equilibrium from Example~\ref{ex:DPE} for $k=2$ and $n=3$. Circled items are allocated to the same player. Each item is named after the value used in the upper bound proof.}\label{vb8}
\end{figure}

To see that this actually yields an $\alpha$-approximate $k$-collusion Nash equilibrium, consider any coalition $K$ of $k$ players. If players play strategy profile $S$, any player in $K$ has utility $n-1$. By switching to the strategy chosen in $\opt$, each player in $K$ obtains utility $(k-1)1+1((n-1)(\alpha-1)+n-k)=\alpha(n-1)$, which is fine. If some players in $K$ choose a subset of the items chosen in $\opt$, and other players in $K$ choose a subset of the items chosen in $S$, then this yields a total value at most $\alpha(n-1)$ for each player. We see that no coalition of $k$ players can improve by deviating, from which the result follows. \qed
\end{proof}
For $\alpha=1$, we obtain the following as a special case.

\begin{corollary}
The $k$-collusion $\poa=1+\frac{n-k}{n-1}$ for set packing games.
\end{corollary}

Although the $k$-collusion price of anarchy is strictly lower than the price of anarchy for all $k\geq 2$, note that this improvement becomes negligible for large $n$. Interestingly, as opposed to all other lower bound examples in this paper, we did not find a matching lower bound example for throughput scheduling games.

\section{Conclusions}
An obvious departure from the suggested model for set packing games, and interesting direction for future research is a model where more than one player may select one and the same item. Then, the actual allocation might be fractional, or probabilistically. This would also allow to consider mixed strategies, which would not be well-defined in the discrete setting we address here.
This departure, however, requires other techniques, as it means to give away the  main distinguishing feature of set packing games that we have exploited to obtain our bounds, namely that players only affect each other via the set of ``available'' items, and that given, do not affect the resulting payoffs.

\section*{Acknowledgements}
A preliminary version with parts of the results presented in this paper appeared in the conference proceedings \cite{deJong2013}. Thanks to Rudolf M\"uller, Frits Spieksma, and Johann Hurink for some helpful discussions.


\begin{thebibliography}{100}

\bibitem{Andelman2009} N.\ Andelman, M.\ Feldman, and Y.\ Mansour. Strong price of anarchy. Games and Economic Behavior, vol 65, 289-317, 2009.

%
\bibitem{Angelucci}
A. Angelucci, V. Bil\`{o}, M. Flammini, and L. Moscardelli. On the sequential price of anarchy of isolation games. In: Z.\ Cai, A.\ Zelikovsky, A.G. Bourgeois (eds.), Computing and Combinatorics (COCOON 2013), Lecture Notes in Computer Science 8591, 17--28, 2013.
%
\bibitem{Anshelevich}
E.\ Anshelevich, A.\ DasGupta, J.\ Kleinberg, {\'{E}}.\ Tardos, T.\ Wexler, and T.\ Roughgarden. The price of stability for network design with fair cost allocation. In {\em Proceedings of the 45th FOCS}, 295--304, 2004.
%

\bibitem{Aumann1959} R.J.\ Aumann. Acceptable points in general Coorpertaive $n$-Person Games. In: 
Contributions to the Theory of Games IV, Annals of Mathematics Studies 40 (R.D. Luce and A. W. Tucker, eds.) 287--324, Princeton University Press, 1959. 

%
\bibitem{AwerbuchSTOC2005}
B.\ Awerbuch, Y.\ Azar, and A.\ Epstein. The price of routing unsplittable flow. In {\em Proceedings of the 37th STOC}, 57--66, 2005.
%
\bibitem{BakkerEtAl2010} V.\ Bakker, M.G.C.\ Bosman, A.\ Molderink, J.L.\ Hurink and G.J.M.\ Smit. Demand side load management using a three step optimization methodology. In {\em Proceedings of the 1st SmartGridComm}, 431--436, 2010.
%
%
\bibitem{BarNoy2001} A.\ Bar-Noy, S.\ Guha, J.\ Naor, and B.\ Schieber. Approximating the throughput of multiple machines in real-time scheduling. SIAM Journal on Computing, vol 31, 331--52, 2001.
%
 
%
\bibitem{BermanDasgupta2000} P.\ Berman and B.\ DasGupta. Multi-phase Algorithms for Throughput Maximization for Real-Time Scheduling. Journal of Combinatorial Optimization, vol 4, 307--323, 2000.
%
%


\bibitem{Bilo}
V.\ Bil{\`{o}}, M.\ Flammini, G.\ Monaco, and L.\ Moscardelli. Some anomalies of farsighted strategic behavior. In: T.\ Erlebach and G.\ Persiano (eds.), Approximation and Online Algorithms (WAOA 2012), Lecture Notes in Computer Science 7846, 229--241, 2013.
%
\bibitem{Bilo2011}
V.\ Bil{\`{o}}, M.\ Flammini, G.\ Monaco, and L.\ Moscardelli. On the performances of Nash equilibria in isolation games. Journal of Combinatorial Optimization, vol 22, 378-391, 2011.
%

\bibitem{Brucker2004}
P.\ Brucker. Scheduling Algorithms (4th ed.), Springer Verlag, Berlin, 2004

\bibitem{Buyya} R.\ Buyya, R.\ Ranjan, and R.N.\ Calheiros. Intercloud: Utility-oriented federation of cloud computing environments for scaling of application services. In: C.H.\ Hsu, L.T.\ Yang, J.H.\ Park and S.S.\ Yeo (eds.), Algorithms and Architectures for Parallel Processing, Lecture Notes in Computer Science 6081, 13--31, 2010.
%

\bibitem{Chandra} B.\ Chandra, and M.M.\ Halld{\'o}rsson. Greedy local improvement and weighted set packing approximation. Journal of Algorithms, vol 39, 223--240, 2001.
%
\bibitem{Christodoulou2005}
G. Christodoulou and E. Koutsoupias. The price of anarchy of finite congestion games. In {\em Proceedings of the 37th STOC}, 67--73, 2005.


\bibitem{deJong2015} J.R.\ Correa, J.\ de~Jong, B.\ de~Keijzer, and M.J.\ Uetz. The Curse of Sequentiality in Routing Games. In: E.\ Markakis and G.\ Sch\"afer:
Web and Internet Economics (WINE 2015), Lecture Notes in Computer Science 9470, 258--271, 2015.


\bibitem{graham1979optimization}
R.~Graham, E.~Lawler, J.~Lenstra, and A.~Rinnooy~Kan.
\newblock {Optimization and approximation in deterministic sequencing and
  scheduling: A survey}.
\newblock {\em Annals of Discrete Mathematics}, 5(2):287--326, 1979.


%
\bibitem{GJ1979} M.R.\ Garey and D.S.\ Johnson, Computers and Intractability: A Guide to the Theory of NP-completeness, Freeman, New York, 1979.

\bibitem{Hassin} R.\ Hassin and U.\ Yovel.\ Sequential scheduling on identical machines. Operations Research Letters, vol 43, 530--533, 2014.
\bibitem{Hayrapetyan} A.\ Hayrapetyan, {\'{E}}.\ Tardos, and T.\ Wexler. The effect of collusion in congestion games. In {\em Proceedings of the 38th STOC}, 89--98, 2006.
%
%
\bibitem{deJong2013}
J.\ de~Jong, M.\ Uetz, and A.\ Wombacher. Decentralized throughput scheduling. In: P.G.\ Spirakis and M.\ Serna (eds.), Algorithms and Complexity (CIAC 2013), Lecture Notes in Computer Science 7878, 134--145, 2013.
%
\bibitem{deJong2014}
J.\ de~Jong and M.\ Uetz. The sequential price of anarchy for atomic congestion games. In: T.-Y.\ Liu, Q.\ Qi and Y.\ Ye (eds.), 
Web and Internet Economics (WINE 2014), Lecture Notes in Computer Science 8877, 429--434, 2014.
%
\bibitem{Karp}
R.M.\ Karp. Reducibility among combinatorial problems. In: Complexity of Computer Computations (R.E.\ Miller, J.W.\ Thatcher, and J.D.\ Bohlinger, eds.), The IBM Research Symposia Series, 85-103, Springer, 1972.
%

\bibitem{Koutsoupias}
E.\ Koutsoupias and C.H.\ Papadimitriou. Worst-case equilibria. Computer Science Review, vol 3, 65--69, 2009. (Preliminary version appeared in Proceedings 16th STACS, 1999)

\bibitem{LRKB1977} J.K.\ Lenstra, A.H.G.\  Rinnooy Kan, and P.\ Brucker. Complexity of Machine Scheduling Problems. Annals of Discrete Mathematics, vol 1, 343--362, 1977.
%

%
\bibitem{Milchtaich1998} I.\ Milchtaich. Crowding Games are Sequentially Solvable. International Journal of Game Theory, vol 27, 501--509, 1998.
%
\bibitem{MolderinkEtAL2010} A.\ Molderink, V.\ Bakker, M.G.C.\ Bosman, J.L.\ Hurink, and G.J.M.\ Smit. Management and control of domestic smart grid technology. IEEE transactions on Smart Grid, vol 1, 109--119, 2010.

\bibitem{Moore1968}
J.M.\ Moore. An $n$ job, one machine sequencing algorithm for minimizing the
number of late jobs. Management Science, 15:102-109, 1968.


\bibitem{Nisan} N.\ Nisan and A.\ Ronen. Algorithmic mechanism design. Games and Economic Behavior, vol 35, 166-196, 2001.
%


\bibitem{PaesLeme} R.\ Paes Leme, V.\ Syrgkanis, and {\'{E}}.\ Tardos. The curse of simultaneity. In {\em Proceedings of the 3rd ITCS}, 60--67, 2012.
%

\bibitem{Papadimitriou2001} C.H. Papadimitriou, Algorithms, Games, and the Internet. In {\em Proceedings of the 33rd Annual ACM Symposium on the Theory of Computing}, 2001, pp. 749--753.


\bibitem{Peters} H.\ Peters, Game Theory: A Multi-Leveled Approach. Springer, 2nd ed., 2015.

\bibitem{Roughgarden} T.\ Roughgarden and {\'E}.\ Tardos. How bad is selfish routing? Journal of the ACM, vol 49, 236-259, 2002.
%

\bibitem{Selten65}
R.\ Selten. Spieltheoretische {B}ehandlung eines {O}ligopolmodells mit {N}achfragetr\"agheit: {T}eil 1: {B}estimmung des dynamischen {P}reis\-gleich\-gewichts. Zeitschrift f{\"u}r die gesamte Staatswissenschaft, vol 121, 301--324, 1965.
%
\bibitem{SkopalikVoecking2008} A.\ Skopalik and B.\ V{\"o}cking. Inapproximability of pure Nash equilibria. In {\em Proceedings of the\ 40th STOC}, 355--364, 2008.
%


\bibitem{tamyca} {\tt http://www.tamyca.com}

\end{thebibliography}
\end{document}